\documentclass[preprint,12pt]{elsarticle}

\usepackage{amsmath}
\usepackage{amsfonts}
\usepackage{amssymb}
\usepackage{amscd}
\usepackage{amsthm}
\usepackage{mathtools}

\usepackage{bbm}
\usepackage{bbding}
\usepackage{mathrsfs}
\usepackage{pifont}
\usepackage{stmaryrd}
\usepackage{txfonts}

\usepackage{array}
\usepackage{booktabs}
\usepackage{diagbox}
\usepackage{threeparttable}

\usepackage{graphicx}
\usepackage{float}
\usepackage{stfloats}
\usepackage{subfig}         
\usepackage{pgfplots}

\usepackage{textcomp}
\usepackage{url}
\usepackage{verbatim}       
\usepackage{color}
\usepackage[left,pagewise]{lineno}
\usepackage[justification=centering]{caption} 

\usepackage{algorithmic}
\usepackage{algorithm}
\usepackage[backref]{hyperref}

\pgfplotsset{compat=1.17}

\newtheorem{theorem}{Theorem}[section]
\newtheorem{lemma}{Lemma}[section]
\newtheorem{definition}{Definition}[section]
\newtheorem{remark}{Remark}[section]
\newtheorem{corollary}{Corollary}[section]
\newtheorem{example}{Example}[section]

\usepackage{bm}

\journal{}

\begin{document}

\begin{frontmatter}
\title{An Efficient and Perfect Secret Sharing Scheme on a Class of Non-maximal Quantum Access Structure}

\author{Lei Li\corref{cor1}}
\cortext[cor1]{Corresponding author.}
\ead{lilei02@xidian.edu.cn}

\author{Zhi Li}
\ead{zhli@xidian.edu.cn}

\address{School of Mechano-Electronic Engineering,
	Xidian University,
	Xi'an, 710071, China}

\begin{abstract}
It is known that existing quantum secret sharing (QSS) schemes constructed on non-maximal quantum access structures (QAS) generally suffer from excessive quantum resource overhead and fail to achieve a favorable trade-off between security and operational efficiency. However, dedicated construction strategies for high-efficiency QSS schemes adaptive to non-maximal QAS have not been fully explored so far. This paper explores construction methods for efficient QSS schemes based non-maximal QAS. We thoroughly analyze three-hyperedge hyperstar QAS and design a universal, efficient, perfect entangled-state-based QSS scheme.
We first characterize the forbidden and intermediate sets of non-maximal QAS, then derive the necessary and sufficient condition for implementing such QAS with pure-state encoded QSS. Next, we define representative-element QAS as a substructure of three-hyperedge hyperstar QAS. Relying on the matching classical secret sharing scheme for auxiliary support, we construct a complete QSS protocol for the hyperstar QAS. By deploying lightweight quantum resources to representative element access structures, the proposed scheme reduces the difficulty in quantum state preparation and distribution, and enhances its security and resource utilization efficiency.

\end{abstract}
\begin{keyword}
Perfect QSS. Non-maximal QAS. Hyperstar. Efficiency.
\end{keyword}

\end{frontmatter}

\section{Introduction}
Quantum secret sharing (QSS) is an important branch of quantum cryptography. In 1999, Hillery et al. \cite{ref1} proposed the first QSS scheme using Greenberger-Horne-Zeilinger (GHZ) triplet states, whose basic idea was to split an unknown quantum state among multiple participants, who must cooperate to reconstruct the original secret information. In the same year, Cleve et al. \cite{ref2} presented a threshold QSS scheme based on quantum error-correcting codes. Since then, numerous QSS schemes have been successively proposed \cite{ref3,ref4}. As a core foundation for multi-party secure quantum communication and distributed quantum computing, QSS leverages quantum mechanics to implement multi-party secret sharing and establishes a novel framework for information security \cite{ref5}. In QSS schemes, the shared secret falls into two categories: classical information encoded onto quantum states \cite{ref1}, or quantum states themselves \cite{ref2}. This paper focuses on QSS schemes where quantum states act as the shared secret.

In QSS schemes involving shared quantum states, the threshold secret sharing scheme is particularly crucial and is typically denoted by $((k, n))$. Here the set of participants is $P=\{P_1,P_2,\cdots,P_n\}$, and any subset consisting of no fewer than $k$ participants in $P$ forms an authorized set. In 1999, Cleve et al. \cite{ref2} first defined threshold schemes for sharing quantum states and proposed a general construction method based on quantum error correction codes, and they precisely delineated the applicable scope of quantum threshold secret-sharing schemes: when $n<2k$, there must exist a corresponding $((k, n))$ threshold scheme (\cite{ref2}, Theorem 5); This further proves that the strong constraint of the pure-state quantum $((k, n))$ threshold scheme must satisfy $n=2k-1$ (\cite{ref2}, Corollary 9).
In 2005, Imai et al.\cite{ref6} proposed a perfect quantum secret sharing (PQSS) scheme: a group of participants shares an unknown quantum state, and only an authorized subset of participants can recover the original secret; participants in unauthorized subsets cannot obtain any information about the secret. i.e. $2^P = \Gamma \cup \Gamma_F$, where $2^P$ denotes the collection of all subsets of $P$, $\Gamma$ denotes the access structure i.e., the family of all authorized subsets, and $\Gamma_F$ represents the family of forbidden sets,  namely all unauthorized subsets. For a QSS scheme, if the authorized sets and unauthorized sets are mutually complementary, the corresponding access structure is referred to as a \textbf{maximal QAS}\cite{ref7}. It is easy to verify that the $((k, n))$ threshold scheme with $n = 2k-1$ is a typical PQSS scheme, in which the unauthorized set and the authorized set are complementary to each other (\cite{ref2}, Corollary 8).
The advantage of the PQSS scheme is its strong security: no unauthorized subset can obtain any information about the secret. However, it has the following drawback: a large class of QAS cannot be realized by PQSS schemes. For instance, it was proved in \cite{ref2} that when $k<n<2k-1$, no $(k, n)$ threshold scheme can be realized by a PQSS scheme. A ramp QSS scheme was introduced in \cite{ref8} for this type of threshold access structure, which divides the set formed by unauthorized subsets into two categories: forbidden set family $\Gamma_F$  and intermediate set family $\Gamma_I$. In essence, there are a lot of non-threshold QAS which also carry intermediate set \cite{ref9,ref10}. We define an access structure containing intermediate set family as a \textbf{non-maximal QAS}, i.e., $2^P = \Gamma \cup \Gamma_F \cup \Gamma_I$. All participants of the forbidden set cannot reconstruct any information about the secret, whereas all participants of intermediate set generally cannot fully recover the secret but may obtain partial information about it.

QAS constitutes the core theoretical framework of QSS schemes, which directly determine the authority division of participants and secret reconstruction rules. Based on their structural characteristics, they can be categorized into two types: maximal QAS and non-maximal QAS. Among them, non-maximal QAS feature intermediate sets that differ from authorized sets and forbidden sets, leading to more intricate structural mechanisms and stricter security constraints. Meanwhile, they are better suited for complex practical quantum network scenarios and have thus become a prominent research hot spot in the field of QSS in recent years.

At present, research on non-maximal QAS is primarily focused on two core directions, and has produced a series of fundamental results. The first direction concerns the property analysis and precise characterization of intermediate sets within non-maximal QAS. Within this research branch, Gheorghiu et al. \cite{ref11} constructed a hybrid QSS scheme based on non-binary stabilizer codes in 2012, establishing a systematic method to judge authorized sets, forbidden sets and intermediate sets in non-maximal structures. To address the potential risk of information leakage induced by intermediate sets, Zhang et al. [12] proposed quantum strongly secure ramp QSS in 2015, which avoids the theft of core quantum information via intermediate sets by imposing rigorous security constraints. Matsumoto et al. \cite{ref13} continuously refined the structural judgment theory: they constructed ramp QSS based on nested linear code systems in 2018 and clarified the judgment criteria for authorized and unauthorized sets. Later in 2020 \cite{ref14}, they further derived the necessary and sufficient conditions for authorized sets and forbidden sets using stabilizer subspaces, realized quantitative calculation of the amount of information obtainable by intermediate sets, and supplemented and perfected the fundamental theoretical framework of non-maximal access structures. The second direction lies in the construction and optimization of efficient QSS schemes under non-maximal QAS. In 1999, Cleve \cite{ref2} proved the efficiency lower bound of the conventional $((k,n))$ quantum threshold secret sharing scheme, clarified the constraint relationship between the qubit overhead of shares and the size of the secret, laid the theoretical foundation for this research branch, and also pointed out the critical drawback of excessive communication overhead during the secret reconstruction phase in traditional schemes. Current  research on efficiency optimization falls into two major technical lines. The first approach constructs PQSS schemes by optimizing and refining QAS. Anderson et al. \cite{ref15} realized optimized adaptation for non-maximal QAS by reconstructing optimally restricted quantum access substructures. Reference \cite{ref16} combines imperfect ramp secret sharing with classical encryption. an overall PQSS scheme can be obtained. Reference \cite{ref17} adopts quantum one-time pads and classical secret sharing to split each participant’s share into quantum and classical parts, thereby reducing the number of quantum shares. Nevertheless, this work does not directly improve the combinatorial structure of QAS.
The second technical line focuses on boosting communication efficiency in the secret reconstruction phase. Building upon the work of Cleve et al. \cite{ref2}, Senthoor et al. constructed several types of quantum threshold secret sharing schemes with high communication efficiency for secret reconstruction. By combining ramp schemes and standard threshold schemes, these schemes allow flexible selection of an arbitrary number of participants when the number of available participants exceeds the threshold according to practical requirements, which improves the communication efficiency of secret recovery \cite{ref18,ref19,ref20}. References \cite{ref21,ref22} established communication-efficient QSS schemes based on extended CSS codes.

A review of existing literature reveals that most research achievements on non-maximal QAS focus on ramp threshold access structures, for which the theoretical characterization and scheme construction systems have become relatively mature. Nevertheless, existing studies still suffer from obvious deficiencies in scenario adaptability as well as theoretical gaps. Practical quantum networks are constrained by transmission distance, channel loss, terrain environment and hardware performance, resulting in asymmetric node connections. Conventional symmetric ramp threshold access structures can hardly adapt to real network scenarios. In comparison, star-type and hyperstar quantum access structures enable direct encoding and modeling for central network nodes, branch nodes and cooperative groups. Such structures can precisely adapt to asymmetric quantum network architectures and possess higher practical engineering value. Chen et al. \cite{ref23} first constructed star-type cluster states, which have been utilized to generate two-dimensional and three-dimensional cluster states \cite{ref23}, topological one-way computation \cite{ref24, ref25}, and QSS\cite{ref26}. A hyperstar structure can directly encode central nodes, branch nodes and cooperative groups within a network into authorized hyperedges, rendering it more practical for quantum networking applications. Reference [27] investigated a restricted hyperstar access structure yet failed to conduct systematic research on hyperstar structures from a unified perspective. Reference [28] constructed multipartite quantum states and quantum information masking schemes via orthogonal arrays to realize a special class of hyperstar access structures integrated with two-way identity authentication. Reference [29] classified all hyperstar with three hyperedges QAS into four isomorphism classes, calculated the optimal information rates of their corresponding classical secret sharing schemes, and further proposed efficient single-qubit schemes for these structures. In summary, existing research on hyperstar access structures remains fragmented and lacks systematic structural analysis and theoretical characterization. More importantly, hyperstar with three hyperedges QAS  constitute a special category of non-maximal QAS that neither fall into the conventional ramp-type access structures nor can be generated by stabilizer codes. As a result, all existing efficient construction methods tailored for ramp structures and stabilizer code systems cannot be directly applied to them. Distinct research gaps still exist regarding property analysis, security characterization and efficient scheme design for such structures.

To address the aforementioned deficiencies in existing research, this paper investigates the theoretical characterization and scheme construction of non-maximal hyperstar QAS. First, we present complete and quantitative characterizations of forbidden sets and intermediate sets within non-maximal QAS, and establish a theoretical framework that can accurately judge subset attributes and quantify the information features of sets. Relying on this framework, we rigorously verify the existence of intermediate sets in hyperstar quantum access structures with three hyperedges. We further optimize and iterate the original hyperstar access structures and propose a novel sub-access structure composed of representative elements. On this basis, combined with classical secret sharing theory, we design an efficient hybrid QSS scheme tailored to the proposed structure. Although the quantum component of our proposed schemes adopts an imperfect structure, the supporting classical secret sharing scheme satisfies perfectness and achieves the optimal information rate. The overall system can form a PQSS scheme. This work provides a universal approach for constructing efficient and perfect secret sharing technologies under non-maximal hyperstar-type QAS.

The organization of this paper is arranged as follows. Section 2 introduces relevant definitions and existing conclusions of QAS, including non-maximal QAS, optimally restricted QAS and hyperstar access structures. Meanwhile, it presents characterizations of forbidden sets and intermediate sets within non-maximal QAS, and further derives a necessary and sufficient condition for a non-maximal quantum access structure to be realized by a pure-state encoded QSS scheme. Section 3 deduces the relevant properties of quantum entangled states. Section 4 constructs four types of perfect hybrid QSS schemes based on hyperstar QAS and provides correctness proofs for the proposed schemes. Section 5 summarizes all research findings of this paper.
\section{Quantum access structure}
 We use the symbol $P=\{P_1,P_2,\cdots,P_n\}$ to denote a set consisting of $n$ participants. For every $A\subseteq P$, write $ \overline{A}=P\setminus A$. For convenience, we also denote $\{P_1,P_2,\cdots,P_n\}$ as $\{1,2,\cdots,n\}$. Furthermore, we write $\{1,2,\cdots,n\}$ as $[n]$, and write $\{m,m+1,\cdots,n\}$ as $[m,n]$.

\begin{definition}
\cite{ref30} Let $\Gamma$ be an access structure. A set $B\in \Gamma$ is called a minimal authorized subset if for every set $A$ satisfying $A\subseteq B$ and $A\neq B$, we have $A \notin\Gamma$. The set consisting of all minimal authorized subsets is called the minimal access structure, denoted by $\Gamma_0$.
\end{definition}

\subsection{Non-maximal QAS}
When the shared secret is a quantum state, to avoid violating the quantum no-cloning theorem, a given access structure is a QAS if and only if the intersection of any two authorized sets in the structure is nonempty \cite{ref2}. Owing to the special nature of QAS, Ref. \cite{ref31} further partitions, for a given QAS $\Gamma$, the collection of unauthorized subsets $\mathcal{A}=\{A\in2^P|A\notin\Gamma\}$ and lets
	\begin{align*}
		\mathcal{A}_1 &= \{ A \in \mathcal{A} \mid \exists B \in \Gamma, A \cap B = \varnothing \} \\
		\mathcal{A}_2 &= \{ A \in \mathcal{A} \mid \forall B \in \Gamma, A \cap B \neq \varnothing \}
	\end{align*}
Obviously, $\mathcal{A}=\mathcal{A}_1\bigcup\mathcal{A}_2$  and $\mathcal{A}_1\bigcap\mathcal{A}_2=\varnothing$. It is clear that all elements in $\mathcal{A}$ are unauthorized subsets. On the other hand, it follows from the above introduction that $\mathcal{A}=\Gamma_F\bigcup \Gamma_I$, where $\Gamma_F \bigcap \Gamma_I=\varnothing$.

\begin{lemma}\cite{ref31}
	In a pure-state QSS scheme, let $\Gamma\subseteq 2^P$ be a QAS, and $\mathcal{A}=\mathcal{A}_1\bigcup\mathcal{A}_2$  be the set consisting of all unauthorized subsets.

	(1) If $A\in \mathcal{A}_1 $, then $\overline{A}\in\Gamma$;

	(2) If $A\in \mathcal{A}_2$, then $\overline{A}\in \mathcal{A}_2$.
\end{lemma}
	
\begin{lemma}\cite{ref31}
	Let $\Gamma\subseteq 2^P$ be a QAS. If $\mathcal{A}_2\neq \varnothing$, there does not exist a PQSS scheme realizing the QAS $\Gamma$.
\end{lemma}
	
	From Lemma 2.2 and Ref.\cite{ref31}, if there exists a PQSS scheme that realizes the QAS $\Gamma$, the scheme is perfect if and only if $\mathcal{A}_2=\varnothing$.

	Ref. \cite{ref8} further divides the unauthorized subsets in the $((k, n))$ quantum threshold access structure into forbidden set family $\Gamma_F$ and intermediate set family $\Gamma_I$. Based on authorized subset, forbidden set  and intermediate set, the $((k, n))$ threshold QSS is classified into two types: perfect $((k,n))$ threshold QSS schemes and ramp $((k,n))$ threshold QSS schemes.
	
	\begin{definition}\cite{ref8}
		A $((k, n))$ QSS scheme is called a $((k,L, n))$-threshold ramp QSS scheme if it satisfies the following conditions:
		\begin{enumerate}
			\item[(i)] A participant subset $B\subseteq P$ belongs to forbidden set $\Gamma_F$ if and only if $|B|\leqslant k-L$;
			\item[(ii)] A participant subset $B\subseteq P$ is an authorized set if and only if $|B|\geqslant k$;
		\end{enumerate}
		The above conditions further imply:
		\begin{enumerate}
			\item[(iii)] A participant subset $B\subseteq P$ belongs to a intermediate set $\Gamma_I$ if and only if $k-L<|B|<k$.
		\end{enumerate}
     \end{definition}
When $L=1$, the $((k,L,n))-$ ramp threshold QSS scheme degenerates into the $((k, n))$ perfect threshold QSS scheme in Ref.\cite{ref2}. Ref.\cite{ref8} proves that the $((k,L,n))-$ ramp scheme can be constructed in pure-state form when $n=2k-L$.

From Definition 2.2, if $B$ is an intermediate set, then its complement $\overline{B}$ also satisfies $k-L<|\overline{B}|<k$. Hence $\overline{B}$ is also an intermediate set; furthermore, if $B$ is a forbidden set if and only if its complement $\overline{B}$ is also an authorized set. Next, we prove that this conclusion also holds for general non-maximal QAS.

\begin{theorem}\label{lem:intermediate-complement}
	For non-maximal QAS $\Gamma$ with $2^P = \Gamma \cup \Gamma_F \cup \Gamma_I$, $\Gamma$ can be realized via a pure-state QSS scheme. Then:\\

  (i) for a set $A$, $A \in \Gamma_I$ if and only if $ \overline{A}\in \Gamma_I$;\\

  (ii) for a set $A$, $A \in \Gamma_F$ if and only if $ \overline{A}\in \Gamma$.

\end{theorem}

\begin{proof}
	Let the quantum state shared by the participants $P$  be $\rho^K$, where $\rho^K$ belongs to a finite-dimensional Hilbert space. Let $R$ be the the purifying ancilla corresponding to the secret quantum state $\rho^K, I(R:A)$ denote the quantum mutual information between the system $R$ and $A$, where $A\subseteq P$.
	
\begin{align}
	I(R:A) &= S(R) + S(A) - S(RA), \label{eq:mutual-A} \\
	I(R:\overline{A}) &= S(R)+S(\overline{A})- S(R\overline{A}).
\end{align}
    Since $RA\overline{A}$ is a purifying quantum system, $S(RA)=S(\overline{A}),S(R\overline{A})=S(A)$.
	Adding the above two equations yields:
	\begin{equation}
		I(R:A) + I(R:\overline{A}) = 2S(R).
	\end{equation}
	Because the reduced state of $R$ is maximally mixed, $S(R)=S(\rho^K)=\log d_K$.  Hence
\begin{equation}
		I(R:A) + I(R:\overline{A}) = 2\log d_K.
\end{equation}

	We first prove statement (i). Suppose that $A \in \Gamma_I$. By definition, $0 < I(R:A) < 2\log d_K$. Using (4) we obtain
\begin{equation}
		0 < I(R:\overline{A}) < 2\log d_K.
\end{equation}
Therefore, $\overline{A} \in \Gamma_I$.
Applying the same argument with $A$ replaced by $\overline{A}$ proves the converse implication. Thus $A \in \Gamma_I$ if and only if $ \overline{A}\in \Gamma_I$.

We next prove statement (ii). Suppose that $A \in \Gamma_F$ . Then $I(R:A)=0.$ Eq. (4) implies $I(R:\overline{A}) = 2\log d_K.$ Therefore, $\overline{A}\in \Gamma $.
Conversely, suppose that $\overline{A}\in \Gamma $. Then $I(R:\overline{A}) = 2\log d_K.$. Again using (4), we obtain $I(R:A)=0$. and hence $A \in \Gamma_F$.Therefore,
$A \in \Gamma_F$ if and only if $ \overline{A}\in \Gamma$.
\end{proof}

In what follows, we prove that the converse proposition of Theorem 2.1 holds true as well.
\begin{theorem}
For non-maximal QAS $\Gamma$ with $2^P = \Gamma \cup \Gamma_F \cup \Gamma_I$, $\Gamma$ satisfies the following two conditions: (i) $B\in\Gamma_I$ if and only if $\overline{B}\in\Gamma_I$; (ii) $ B\in\Gamma_F$ if and only if $\overline{B}\in\Gamma$. Then there exists a pure-state isometric encoding

\[
\phi: \mathcal{H}_K \longrightarrow \bigotimes_{i \in P} \big(\mathcal{H}_i^{(1)} \otimes
                                                   \mathcal{H}_i^{(2)}\big)
\]
so that the following conditions hold:

(i) $A \in \Gamma$ if and only if $A$ can exactly recover the secret $|K\rangle$;

(ii) $A \in \Gamma_F$ if and only if $A$ contains zero information regarding the secret $|K\rangle$;

(iii) $A \in \Gamma_I$ if and only if $A$ carries partial information but cannot reconstruct the secret $|K\rangle$.
\end{theorem}

\begin{proof}
Let the intermediate set be $\Gamma_I$=$\{B_1,\overline{B_1},B_2,\overline{B_2},\dots,B_m,\overline{B_m }\}$. Add $B_1$ and $\overline{B_1}$ to $\Gamma $ separately to yield two new access structures, denote them by $ \Gamma_ {B_1}$  and $ \Gamma_{\overline{B_1}}$, and check whether they are maximal QAS. It can be readily verified that the two access structures are either both maximal QAS, or neither of them is a maximal QAS. If both are maximal QAS, then we have that $\Gamma= \Gamma_ {B_1} \cap \Gamma_ {\overline{B_1}}$; if neither of them is a maximal QAS, add $B_2$ and $\overline{B_2}$ to $ \Gamma_ {B_1}$ and $ \Gamma_ {\overline{B_1}}$, respectively, separately to yield two new access structures, denote them by $\Gamma_ {B_1,B_2}$  and $\Gamma_{\overline{B_1},\overline{B_2}}$, and check whether they are maximal QAS.If both are maximal QAS, then we have that $\Gamma= \Gamma_ {B_1,B_2} \cap \Gamma_{\overline{B_1},\overline{B_2}}$. If neither of them is a maximal QAS, repeating the above procedure, we eventually get $\Gamma= \Gamma_ {B_1,B_2,\dots,B_r} \cap \Gamma_ {\overline{B_1},\overline{B_2},\dots,\overline{B_r}}$, where $\Gamma_ {B_1,B_2,\dots,B_r}$ and $ \Gamma_ {\overline{B_1},\overline{B_2},\dots,\overline{B_r}}$ are maximal QAS with $r\in [1,m]$. Furthermore, from the construction of these two maximal QAS, it follows that $A$ can only belong to one of the two if $A \in \Gamma_I $ .

Since perfect pure-state QSS schemes exist for maximal QAS, there accordingly exist pure-state isometric encodings for $\Gamma_ {B_1,B_2,\dots,B_r}$ and $ \Gamma_ {\overline{B_1},\overline{B_2},\dots,\overline{B_r}}$, respectively.

For $\Gamma_ {B_1,B_2,\dots,B_r}$ , there exists a pure-state isometric encoding
\[\phi_1: \mathcal H_{K_1 }\to \bigotimes_{i \in P} \mathcal H_i^{(1)}\]

For $ \Gamma_{\overline{B_1},\overline{B_2},\dots,\overline{B_r}}$, there exists a pure-state isometric encoding
 \[\phi_2:\mathcal H_{K_2 }\to \bigotimes_{i \in P} \mathcal H_i^{(2)}\]
so that $ A$ can reconstruct the secret $ K_1$  when $A\in \Gamma_{B_1,B_2,\dots,B_r}$; $ A$ can reconstruct the secret $K_2$  when $A \in  \Gamma_ {\overline{B_1},\overline{B_2},\dots,\overline{B_r}}$; and $ A$ cannot reconstruct the secret $ K_1$  when $A\notin \Gamma_{B_1,B_2,\dots,B_r}$; $ A$ cannot reconstruct the secret $K_2$  when $A \notin  \Gamma_ {\overline{B_1},\overline{B_2},\dots,\overline{B_r}}$.

Let the dimension of the secret space $\mathcal H_K$ be $d \geq 2$, and fix an orthonormal basis $\{|x\rangle \,\colon\, x = 0,\,1,\,\dots,\,d - 1\}$. Define the isometry
\[
\tau:\mathcal H_{K
}\to \mathcal{H}_{K_1} \otimes  \mathcal{H}_{K_2}
\]
satisfies
\[
\tau |x\rangle = |x\rangle_{K_1} \otimes |x\rangle_{K_2}
\]
So
\[
\tau\left( \sum_{x} \alpha_x |x\rangle \right) = \sum_{x} \alpha_x |x\rangle^{\otimes 2}
\]
Let \[
\phi = \left( \bigotimes_{j=1}^2 \phi_j \right) \tau
\]

If $A \in \Gamma $, then $A \in \Gamma_{B_1,B_2,\dots,B_r}$ , and $A \in \Gamma_ {\overline{B_1},\overline{B_2},\dots,\overline{B_r}}$, thus $A$ can reconstruct ${K_1}$ and ${K_2}$. By further applying the inverse isometry of the outer encoding $\tau$, the original secret $|K\rangle$ can be recovered; if $A \in \Gamma_I$, then $A$ belongs to only one of $ \Gamma_{B_1,B_2,\dots,B_r}$ and $\Gamma_ {\overline{B_1},\overline{B_2},\dots,\overline{B_r}}$, thus $A$ can reconstruct exactly one of $|K_1\rangle$ and $|K_2\rangle$, i.e., carry partial information but cannot reconstruct the secret $|K\rangle$; if $A \in \Gamma_F$, then $A$  belongs to neither of $ \Gamma_{B_1,B_2,\dots,B_r}$ and $\Gamma_ {\overline{B_1},\overline{B_2},\dots,\overline{B_r}}$, thus $A$ cannot reconstruct either of $|K_1\rangle$ and $|K_2\rangle$, so $A$ contains zero information regarding the secret $|K\rangle$.

Since $\tau$ and each $\phi_i$ are  isometries, $ \phi $ is also an isometry. Therefore, this scheme is a pure-state scheme.

\end{proof}

\begin{theorem}
For non-maximal QAS $\Gamma$ with $2^P = \Gamma \cup \Gamma_F \cup \Gamma_I$, and $\Gamma$ can be realized via a pure-state QSS scheme. Then $\Gamma$ has the following properties:

(i) $\Gamma_I=\mathcal{A}_2$;

(ii) $\Gamma_F=\mathcal{A}_1$.
\end{theorem}

\begin{proof}
	We first prove that $\mathcal{A}_2\subseteq \Gamma_I$. Since $\mathcal{A}_2\neq \varnothing$, take $A\in \mathcal{A}_2$. By the definition of $\mathcal{A}_2$, we have $A\notin\Gamma$. Next we prove that $A\notin\Gamma_F$. Assume $A\in \Gamma_F$, i.e., $A$ is a forbidden set. By Theorem 2.1, which implies $\bar{A}\in\Gamma$. However,  Lemma 2.1 gives $\bar{A}\in \mathcal{A}_2$, leading to a contradiction. Thus the assumption is invalid, so $A\notin \Gamma_F$. Consequently $A\in \Gamma_I$.
		
	Then we prove that $\Gamma_I\subseteq \mathcal{A}_2$. Let $A \in \Gamma_I$. By the definition of an intermediate set, the elements of an intermediate set are first unauthorized subsets, thus $A \in \mathcal{A}$. Next, we prove $A \notin \mathcal{A}_1$. If $A \in \mathcal{A}_1$, then by Lemma 2.1, $\overline{A} \in \Gamma$. Further, by Theorem 2.1, $A \in \Gamma_F$, which contradicts $A \in \Gamma_I$. Thus, we have $A \in \mathcal{A}_2$, which implies $\Gamma_I \subseteq \mathcal{A}_2$.
	
By combining $\mathcal{A}_1 \cup \mathcal{A}_2=\Gamma_I\cup \Gamma_F$ and (i), (ii) can be derived.
\end{proof}
	
Theorem 2.3 presents a computational approach to calculate intermediate sets and forbidden sets for pure-state QSS schemes involving intermediate sets.	

Let $t_{\max}^{(1)}=max_{A\in\mathcal{A}_1}|A|$, $t_{\min}^{(2)}=min_{A\in\mathcal{A}_2}|A|$. For the $((k,n))$ ramp threshold QAS, it is obvious that $t_{\min}^{(2)}=t_{\max}^{(1)}+1$.
	
\begin{corollary}
		For the ramp quantum threshold access structure $((k,L,n))$ with $n=2k-L$, and $1<L<k$, we have
\begin{align}
		t_{\min}^{(2)}=k-L+1.
\end{align}
\end{corollary}

\begin{proof}
		By Theorem 2.2 we have
    \[
    \begin{aligned}
	\mathcal{A}_1 &= \bigl\{ A \in \mathcal{A} \bigm| |A| \leqslant k-L \bigr\}, \\
	\mathcal{A}_2 &= \bigl\{ A \in \mathcal{A} \bigm| k-L < |A| \leqslant k-1 \bigr\}.
    \end{aligned}
    \]
An element $A$ in $\mathcal{A}_2$ is an unauthorized subset that intersects with all authorized subsets containing $k$ participants. The number of such unauthorized subsets is at least $t_{\min}^{(2)}$, hence $t_{\min}^{(2)} = k - L + 1$.

\end{proof}

From Corollary 2.1 and the relation $t_{\min}^{(2)} = t_{\max}^{(1)} + 1$, we obtain $t_{\max}^{(1)} = k - L$.

	\subsection{Optimal restricted QAS}
	\begin{definition}\cite{ref15}
		The restriction of a QAS $\Gamma=\{A_1,A_2,\dots,A_r\}$ realizable on a subset $B\subseteq P$ is defined as
		\[
		\Gamma|_B=\{A_i\cap B: A_i\in\Gamma\},
		\]
		where $\Gamma|_B$ satisfies the no-cloning theorem, and $B\cap A_i\neq \varnothing$ holds for any $A_i\in\Gamma$.
		If $\Gamma|_B$ cannot be further improved, it is called minimal. If there does not exist another subset $D\subseteq P$ with $|D|<|B|$ such that $\Gamma|_D$ is minimal, then $\Gamma|_B$ is said to be optimal.
	\end{definition}
	
	\begin{remark}
		When discussing the improved QAS, access structures of the forms $\Gamma|_B=\{12\}$ and $\Gamma|_B=\{4\}$, i.e., minimal access structures containing exactly one minimal authorized subset, suffer degraded security. Therefore, we stipulate in this paper that such cases are excluded from the improved QAS.
	\end{remark}

\begin{example}
	Let the minimal QAS be $\Gamma_0=\{12,134,135,2345\}$. Four qualified subsets $B$ can be found, denoted by $B_1,B_2,B_3,B_4$, where
	\[
	B_1=\{123\},\; B_2=\{1234\},\; B_3=\{1235\},\; B_4=\{1245\}.
	\]
	Then
	\[
	\begin{aligned}
		\Gamma_0|_{B_1}&=\{12,13,23\},\\
		\Gamma_0|_{B_2}&=\{12,13,234\},\\
		\Gamma_0|_{B_3}&=\{12,13,235\},\\
		\Gamma_0|_{B_4}&=\{12,14,15,245\}.
	\end{aligned}
	\]
	By Definition 2.3, $\Gamma_0|_{B_1}$ is an optimal restricted access structure.
\end{example}
	
\subsection{Hyperstar QAS and its improvement}
This section proposes an improved QAS corresponding to hyperstar QAS, namely the QAS composed of representative elements.
	
\subsubsection{Hyperstar QAS}
	\begin{definition}\cite{ref30}
		Let $V=\{v_1,v_2,\dots,v_n\}$ be a finite set. A hypergraph $H=\{E_1,E_2,\dots,E_m\}$ on $V$ is a finite family of subsets of $V$ satisfying:
		\begin{enumerate}
			\item[(a)] $E_i\neq \varnothing (i=1,2,\dots,m)$.
			\item[(b)] $ \bigcup_{i=1}^m E_i=V$.
		\end{enumerate}
		In hypergraph $H$, elements $v_1,v_2,\dots,v_n$ of $V$ are called vertices, and sets $E_1,E_2,\dots,E_m$ are called hyperedges.
		For convenience, we adopt the following convention: Let $H(V,E)$ denote a hypergraph with vertex set $V$ and hyperedge set $E$.
	\end{definition}
	
Let each participant correspond to a vertex of a hypergraph and each minimal authorized subset correspond to a hyperedge of the hypergraph. Then every minimal access structure $\Gamma_0\subseteq 2^P$ on $P$ is in one-to-one correspondence with a hypergraph $H(V,E)$.
	
	\begin{definition}\cite{ref30}
		A hypergraph $H(V,E)$ is called a hyperstar if there exists a hyperedge sequence $E_1,E_2,\dots,E_{m-1}$ in $E$ such that
		\begin{enumerate}
			\item[(1)] $\bigcap_{E_i\in E}E_i\neq \varnothing$;
			\item[(2)] For any $i$, there exists $v\in E_i$ satisfying $v\notin E_j$ for all $i\neq j$.
		\end{enumerate}
	\end{definition}

It follows from Definition 2.5 that if a hyperstar $H(V,E)$ contains only one hyperedge, i.e., $E=\{E_1\}$, it can be realized by the $(|E_1|,|E_1|)$ Shamir threshold secret sharing scheme. Its optimal information rate satisfies $\rho^*(E)=1$, which means the access structure is ideal \cite{ref30}. Since a hyperstar access structure with two hyperedges is a hyperpath, the optimal information rate $\rho^*(\Gamma_i)$  for the access structures induced by hyperpaths was proven to be $2/3$ [30]. Ref. \cite{ref29} constructs classical secret sharing schemes with optimal information rate for hyperstar with three hyperedges QAS $\Gamma_i$,  where $i\in[4]$.
	
	\begin{lemma}[\cite{ref29}, Theorem 3.1]
		Let $\Gamma_i$ be a hyperstar with 3 hyper-edges. Then there exists a secret sharing scheme on $\Gamma_i$ with optimal information rate $\rho^*(\Gamma_i) = 1$,
 $i \in [2]$.
	\end{lemma}
	
	\begin{lemma}[\cite{ref29}, Theorem 3.2]
		Let $\Gamma_i$ be a hyperstar with 3 hyper-edges. Then there exists a secret sharing scheme on $\Gamma_i$ with optimal information rate $\rho^*(\Gamma_i) = 2/3$, where $i \in [3,4]$.
	\end{lemma}

\subsubsection{QAS composed of representative elements}
In the construction of QSS schemes, optimal QAS are advantageous in resource consumption reduction, yet they do not necessarily facilitate practical implementation of the scheme. In what follows, we introduce an improved structure for hyperstar QAS, i.e., the QAS constructed from representative elements. Such structures can be obtained by taking the restriction defined in Section 2.2 of the original access structure with respect to the set formed by the representative element of each hyperedge.
For instance, for $\Gamma_1=\{A_1A_2,A_1A_3,A_1A_4\}$, select the first participant $P_1^{(i)}$ from each $A_i$ with $i\in[4]$, then we can obtain
\begin{align}
	B=\{P_1^{(1)},P_1^{(2)},P_1^{(3)},P_1^{(4)}\},
\end{align}
	and the representative access structure of $\Gamma_1$ is
	\[
	\Gamma|_B=\{P_1^{(1)}P_1^{(2)},P_1^{(1)}P_1^{(3)},P_1^{(1)}P_1^{(4)}\},
	\]
	denoted as $R_1$.
	
The representative QAS corresponding to four hyperstar QAS with three hyperedges are listed in Table 1.

\begin{table}[H]
	\caption{The representative QAS corresponding to four hyperstar with three hyperedges QAS .\label{tab2}}
	\resizebox{138mm}{!}{
	\begin{tabular}{c l l}
		\toprule
		&	Hyperstar QAS	&	Representative element QAS	\\
		\midrule
		1	&	$\Gamma_1=\{A_1A_2,A_1A_3,A_1A_4\}$	&	$R_1=\{P_1^{(1)}P_1^{(2)}, P_1^{(1)}P_1^{(3)}, P_1^{(1)}P_1^{(4)}\}$	\\ \hline
		2	&	$\Gamma_2=\{A_1A_2A_4,A_1A_2A_5,A_1A_3\}$	&	$R_2=\{P_1^{(1)}P_1^{(2)}P_1^{(4)}, P_1^{(1)}P_1^{(2)}P_1^{(5)}, P_1^{(1)}P_1^{(3)}\}$\\ \hline
		3	&	$\Gamma_3=\{A_1A_2A_4,A_1A_2A_3A_5,A_1A_3A_6\}$	&	$R_3=\{P_1^{(1)}P_1^{(2)}P_1^{(4)}, P_1^{(1)}P_1^{(2)}P_1^{(3)}P_1^{(5)}, P_1^{(1)}P_1^{(3)}P_1^{(6)}\}$\\ \hline
		4	&	$\Gamma_4=\{A_1A_2A_4A_7,A_1A_2A_3A_5,A_1A_3A_6A_7\}$	&	$R_4=\{P_1^{(1)}P_1^{(2)}P_1^{(4)}P_1^{(7)}, P_1^{(1)}P_1^{(2)}P_1^{(3)}P_1^{(5)}, P_1^{(1)}P_1^{(3)}P_1^{(6)}P_1^{(7)}\}$\\
		\bottomrule
	\end{tabular}
	}
\end{table}

It is easy to verify that all four representative access structures $R_i\ (i\in[4])$ are valid QAS. For $R_1$, by calculation, we obtain that
	\[
	\mathcal{A}_2^{(1)}=\{P_1^{(1)},P_1^{(2)}P_1^{(3)}P_1^{(4)}\}.
	\]
   \[
	\mathcal{A}_1^{(1)}=\{P_1^{(2)},P_1^{(3)},P_1^{(4)},P_1^{(2)}P_1^{(3)},P_1^{(2)}P_1^{(4)},P_1^{(3)}P_1^{(4)}\}.
	\]
Set $\Gamma_I^{(1)}=\mathcal{A}_2^{(1)}, \Gamma_F^{(1)}=\mathcal{A}_1^{(1)}$, then it is readily verified that $\Gamma$ satisfies the following two conditions: (i) $B\in\Gamma_I$ if and only if $\overline{B}\in\Gamma_I$; (ii) $ B\in\Gamma_F$ if and only if $\overline{B}\in\Gamma$. Thus QAS $R_1$ satisfies the conditions of Theorem 2.2, so it can be realized by a pure-state encoding scheme. However, Lemma 2.2 indicates that this access structure cannot be realized via a PQSS scheme.

Similar computation for all $R_i$ gives $\mathcal{A}_2^{(i)}\neq \varnothing$ for $i\in [4]$. Thus each $R_i$ contains intermediate sets. Combining with Theorem 2.2 we have the following conclusion:
	
\begin{theorem}
		QAS $R_i\ (i\in[4])$ formed by representative elements of hyperstar with three hyperedges QAS is a non-maximal QAS, and this quantum access structure is realizable with a pure-state encoding scheme.
\end{theorem}

Theorem 2.4 shows that after transforming large-participant QAS $\Gamma_i$ into lightweight one $R_i$. However the obtained QAS $R_i$ remains unrealizable with PQSS
by Lemma 2.2.

\section{Properties of quantum entangled states}

\begin{lemma}\cite{ref32} (Schmidt Decomposition Theorem)
	Let $\vert\psi\rangle$ be a pure state of the composite quantum system $AB$.
	Then there exist an orthonormal basis $\{\vert i_A\rangle\}$ of subsystem $A$ and an orthonormal basis $\{\vert i_B\rangle\}$ of subsystem $B$, such that
\begin{align}
	\vert\psi\rangle = \sum_{i=1}^{l} \lambda_i \vert i_A\rangle\vert i_B\rangle,
\end{align}
	where $\lambda_i$ are nonnegative real numbers satisfying
\begin{align}
	\sum_{i=1}^{l} \lambda_i^2 = 1.
\end{align}
	The number of nonzero elements in $\{\lambda_i\}$ is denoted by $k$, which is called the {\em Schmidt number} of $\vert\psi\rangle$.
	Furthermore, $\vert\psi\rangle$ is an entangled state if and only if $k\ge 2$.
\end{lemma}

\begin{theorem}
	Let $\vert\psi\rangle$ be an entangled state on bipartite system $AB$, and let $U_A, U_B$ be unitary operators on subsystems $A$ and $B$, respectively.
	Then $(U_A\otimes U_B)\vert\psi\rangle$ is also an entangled state and has the same Schmidt number as $\vert\psi\rangle$.
\end{theorem}

\begin{proof}
	Suppose the Schmidt number of $\vert\psi\rangle$ is $k$.
	Then there exist orthonormal bases $\{\vert i_A\rangle\}$ and $\{\vert i_B\rangle\}$ of $A$ and $B$, together with $\lambda_i>0$, such that
	\[
	\vert\psi\rangle = \sum_{i=1}^{k} \lambda_i \vert i_A\rangle\vert i_B\rangle,
	\qquad
	\sum_{i=1}^{k}\lambda_i^2 = 1.
	\]
	Since $U_A$ and $U_B$ are unitary transformations,
	$\{U_A\vert i_A\rangle\}$ and $\{U_B\vert i_B\rangle\}$ are still orthonormal bases.
	We have
	\[
	(U_A\otimes U_B)\vert\psi\rangle
	= \sum_{i=1}^{k} \lambda_i \,U_A\vert i_A\rangle\, U_B\vert i_B\rangle.
	\]
	So the Schmidt number of $(U_A\otimes U_B)\vert\psi\rangle$ is also $k$.
\end{proof}

In the following, let $|K\rangle = \alpha\vert0\rangle + \beta\vert1\rangle$ be a single-qubit state.
For simplicity, assume $\alpha>0,\beta>0$ and $\alpha^2+\beta^2=1$.
Consider an $m$-partite entangled state with $m\ge 2$:

\begin{equation}
	\lvert \varphi \rangle = \frac{1}{\sqrt{2}} \big( \lvert 00\cdots0 \rangle + \lvert 11\cdots1 \rangle \big).
	\tag{8}
	\label{eq:ghz-like}
\end{equation}

Then we have
\begin{align}
	|K\rangle\otimes|\varphi\rangle
	&=\left(\alpha|0\rangle+\beta|1\rangle\right)
	(\frac{1}{\sqrt{2}}\left(|00\cdots0\rangle+|11\cdots1\rangle\right) \notag\\
	&=\frac{1}{\sqrt{2}}\bigg(\alpha|00\cdots0\rangle+\alpha|01\cdots1\rangle
	+\beta|10\cdots0\rangle+\beta|11\cdots1\rangle\bigg) \tag{9}
\end{align}

We perform the CNOT operation on the first and second qubits in Eq.~(9), where the first qubit is the control qubit.
Then the state $|K\rangle\otimes|\varphi\rangle$ is transformed into the state $|\psi_1\rangle$, where
\begin{align}
	|\psi_1\rangle
	&=\frac{1}{\sqrt{2}}\bigg(\alpha|000\cdots0\rangle+\alpha|011\cdots1\rangle
	+\beta|110\cdots0\rangle+\beta|101\cdots1\rangle\bigg) \notag\\
	&=\left(\alpha|00\rangle+\beta|11\rangle\right)\frac{1}{\sqrt{2}}|00\cdots0\rangle
	+\left(\alpha|01\rangle+\beta|10\rangle\right)\frac{1}{\sqrt{2}}|11\cdots1\rangle \tag{10}
\end{align}
Here, $\alpha|00\rangle+\beta|11\rangle$ and $\alpha|01\rangle+\beta|10\rangle$ are two orthogonal states in
$C^4=C^2\otimes C^2$, and $|00\cdots0\rangle$ and $|11\cdots1\rangle$
are two orthogonal states in $(C^2)^{\otimes (m-1)}$.
Thus, Eq.~(10) is the Schmidt decomposition of $|\psi_1\rangle$ with Schmidt number 2.
By Lemma 3.1, $|\psi_1\rangle$ is an entangled state.

\section{PQSS Schemes on four types of hyperstar QAS}
\subsection{System definition}
In this section, we construct PQSS schemes for the four types of hyperstar QAS $\Gamma_i~(i\in[4])$ listed in Table 1. We take the representative access structure $R_i$ corresponding to $\Gamma_i$ as the auxiliary structure.
In this model, there exists a distributor Alice, a participant set $\mathcal{P}$, and an eavesdropper Eve over the quantum channel. All participants are assumed to be honest.

It should be noted that $P=\{P_1,P_2,\dots,P_n\}$. On the other hand, since
\[
P=\bigcup_{j=1}^{r}A_j,
\]
we further denote
\[
P=\bigcup_{j=1}^{r}\{P_{1}^{(j)},\dots,P_{n_j}^{(j)}\},
\]
where $n_1+n_2+\dots+n_r=n$, $r$ is the number of subsets of $\Gamma_i$ that contain the non-empty set $ A_j$.

The model belongs to a mixture network structure that is made up of classical access structure $\Gamma_i$ and QAS $R_i$. For simplicity, we use $(\Gamma_i, R_i)$ to represent the structure of this mixed model, where $i\in[4]$.
Furthermore, the classical perfect secret sharing scheme constructed on $\Gamma_j$ is consistent with the schemes described in Lemma 2.3 and Lemma 2.4, thus we assume that the minimal authorized subsets in $\Gamma_i$ have obtained the classical secret key $s$ according to this scheme, where $s\in F_p$, $F_p$ is a finite field, and $p$ is a sufficiently large prime number. The QSS scheme constructed on $R_i$ is implemented based on the scheme over $\Gamma_j,j\in[4] $. Specifically, the secret $s$ reconstructed by the minimal authorized subsets on $\Gamma_i$ is used to encrypt and decrypt the share quantum states on $R_i,i\in[4]$.
We assume that representative element $P_1^{(i)},i\in[r]$, can operate these quantum states via relevant unitary transformations, and has the ability to measure these particle states.

The main distribution rules of the proposed scheme are given as follows:

(1) Alice shares the classical secret $s$ in $\Gamma_i$ by adopting the classical secret sharing schemes in Lemma 2.3 and Lemma 2.4, where each participant $P_i~(i\in[n])$ in $P$ receives the corresponding classical share.
	
(2) Alice selects a classical bit string $s$, which can be regarded as a classical bit string composed of $0$ and $1$. Each pair of classical bits decides to perform the $I$, $X$, $Y$ or $Z$ transformation on the qubit, where $I,X,Y$ and $Z$ are Pauli operators. If the classical bit pairs are $00$, $01$, $10$ and $11$, the $I$, $X$, $Y$ and $Z$ transformations are applied, respectively. The quantum state $|\varphi\rangle$ is encrypted by via the above transformation, and the encrypted quantum state is denoted by $|\tilde{\varphi}\rangle$.

(3) The quantum state $|\tilde{\varphi}\rangle$ is shared in $R_i$ via the QSS scheme, and each participant in the representative element set obtains the corresponding particles from quantum state  $|\tilde{\varphi}\rangle$. It should be noted that we assume that the system configure $\Gamma_E$ detection entangled states $|\varphi\rangle$ defined in Eq. (8) whenever the system transmits a shared quantum state $|\tilde{\varphi}\rangle$, where $|\varphi\rangle$ contains $m$ particles,
 and needs to send the $m+1$ particles in $|\tilde{\varphi}\rangle$ to the participants in the minimal authorized subset. So when Alice distribute the particles of the shared quantum state $|\tilde{\varphi}\rangle$ to participants, these particles are interleaved among the particles from the $\Gamma_E$ detection quantum states.

The main reconstruction rules of the scheme are as follows:

(1) After eavesdropping detection, the participants in the minimal authorized subset collect the quantum shares.

(2) The quantum state $|\widetilde{\varphi}\rangle $ is decrypted with the classical key $s$.

\subsection{PQSS schemes on four types of hyperstar access atructures}
In this section, with reference to the conventions specified in the introduction of this paper, we elaborate on the distribution method of shares for QSS schemes defined over QAS, i.e., representative-element access structures—by specifying how quantum shares are distributed for each authorized subset. The distribution of classical shares adopts the classical secret sharing schemes in Lemma 2.3 and Lemma 2.4 to share the secret $s$ over $\Gamma_i$. It is assumed that the participants in the authorized subset have reconstructed the classical key $s$ with $s\in F_p$. Here we require $s\ge 993$ to ensure that the bit length $l$ of the binary representation $j_1j_2\cdots j_l$ of $s$ satisfies $l\ge 10$.

It should be noted that, in the classical secret sharing scheme over $\Gamma_i$, when Alice selects the secret $s$, it is guaranteed that at least four pairs in its binary representation $j_1j_2\cdots j_l$ are not equal to $00$, i.e. suppose \[
j_{i_1}j_{i_2},\;j_{i_3}j_{i_4},\;j_{i_5}j_{i_6},\;j_{i_7}j_{i_8}
\]
are ordered pairs selected from different positions of $j_1j_2\cdots j_l$, considering the encryption requirement of quantum states, we require that
\[
j_a j_b \in \{01,10,11\},
\]
where $j_a j_b \in\{j_{i_1}j_{i_2},\;j_{i_3}j_{i_4},\;j_{i_5}j_{i_6},\;j_{i_7}j_{i_8}\}$.
\subsubsection{PQSS scheme on $\Gamma_1=\{A_1A_2,A_1A_3,A_1A_4\}$}
First, we present the distribution of quantum shares and the recovery process of the quantum secret state $|K\rangle$ for the minimal authorized subset $A_1A_2$.

Let the quantum secret state be $|K\rangle=\alpha|0\rangle+\beta|1\rangle$, where $\alpha>0$, $\beta>0$, and $\alpha^2+\beta^2=1$. We take the bipartite entangled state for $|\varphi\rangle$ in Eq. (8), i.e.,
\[
|\varphi\rangle=\frac{1}{\sqrt{2}}\big(|00\rangle+|11\rangle\big).
\]
Alice performs the CNOT operation on the first and second qubits of the product state $|K\rangle\otimes|\varphi\rangle$, where the first qubit is the control qubit. Then $|K\rangle\otimes|\varphi\rangle$ evolves into the state $|\psi_1\rangle$, and according to Eq. (10), we have
\[
|\psi_1\rangle
= \left(\alpha|00\rangle+\beta|11\rangle\right)\frac{1}{\sqrt{2}}|0\rangle
+\left(\alpha|01\rangle+\beta|10\rangle\right)\frac{1}{\sqrt{2}}|1\rangle. \tag{11}
\]
It follows from the analysis of Eq. (10) that $|\psi_1\rangle$ is an entangled state.

(1) Distribution of quantum shares

(1.1) Alice selects the bit positions of ordered pairs in the binary representation $j_1j_2\cdots j_m$ of $s$, and secretly informs $P_1^{(1)},P_1^{(2)}$ of the corresponding positions of these ordered pairs. For example, the position of $j_{i_1}j_{i_2}$ is told to $P_1^{(1)}$, and the position of $j_{i_3}j_{i_4}$ is told to $P_1^{(2)}$. Since the representative participants $P_1^{(1)},P_1^{(2)}$ have reconstructed the secret $s$, they can obtain $j_{i_1}j_{i_2}$ and $j_{i_3}j_{i_4}$ from $s$, respectively. According to the system model in this section, $P_1^{(1)}$ and $P_1^{(2)}$ can determine the Pauli operators $\sigma_1$ and $\sigma_2$ from $j_{i_1}j_{i_2}$ and $j_{i_3}j_{i_4}$, respectively, where $\sigma_i\in\{X,Y,Z\},\;i\in\{2\}$.

(1.2) Alice applies the operator $\sigma_1\otimes I\otimes\sigma_2$ sequentially to the first to the third particles (in order from left to right) in $|\psi_1\rangle$, and obtains
\[
|\tilde{\varphi}\rangle = (\sigma_1\otimes I\otimes\sigma_2)|\psi_1\rangle.
\]
She sends the first and second particles of $|\tilde{\varphi}\rangle$ to $P_1^{(1)}$, and the third particle to $P_1^{(2)}$. When transmitting the three particles of $|\tilde{\varphi}\rangle$ to the participants in the minimal authorized subset, it is necessary to prepare $\Gamma_E$ detection entangled states $|\varphi\rangle$ defined in Eq. (8), where each $|\varphi\rangle$ contains two particles. When distributing to participant $P_1^{(i)}, i\in[2]$, the particles of the shared quantum state $|\tilde{\varphi}\rangle$ are interleaved among these  particles $\Gamma_E$ detection entangled states $|\varphi\rangle$.

$P_1^{(1)}$ and $P_1^{(2)}$ measure the particle states at the relevant positions informed by Alice. If their measurement results are identical (both $1$ or both $0$), the protocol proceeds to the next step; otherwise, the process will restart from (1.1).

(2) Reconstruction of quantum secret state $|K\rangle=\alpha|0\rangle+\beta|1\rangle$

\vspace{0.5em}
(2.1) When $P_1^{(1)}$ and $P_1^{(2)}$ receive their particles, respectively, $P_1^{(1)}$ acts on its first received particle with $\sigma_1$, and $P_1^{(2)}$ acts on its received particle with $\sigma_2$. Meanwhile, $P_1^{(2)}$ informs $P_1^{(1)}$ that he has completed the operation.

\vspace{0.5em}
(2.2) $P_1^{(1)}$ performs the CNOT operation on its first particle and second particle, where the first particle serves as the control qubit.

\vspace{0.5em}
The distribution of quantum shares and the recovery process of the quantum secret state $|K\rangle$ for participants on the minimal authorized subsets $A_1A_3,A_1A_4$ are similar to those on $A_1A_2$.

\subsection*{4.2.2 PQSS scheme for $\Gamma_2=\{A_1A_2A_4,A_1A_2A_5,A_1A_3\}$}
The distribution of quantum shares and the recovery process of the quantum secret state $|K\rangle$ for participants on $A_1A_3$ in $\Gamma_2$ are similar to the corresponding case of $A_1A_2$ in $\Gamma_1$. The following presents the distribution of quantum shares and the recovery process of the quantum secret state $|K\rangle$ on the minimal authorized subset $A_1A_2A_4$.

Alice sets \[
| \varphi \rangle = \dfrac{1}{\sqrt{2}}\big(| 000 \rangle + | 111 \rangle\big).
\] for the entangled state $| \varphi \rangle$ in Eq. (8).

Alice performs the CNOT operation on the first and second qubits of the direct product state $| K \rangle \otimes |\varphi \rangle$, where the first qubit is the control qubit. At this time, $| K \rangle \otimes |\varphi \rangle$ evolves into the state $| \psi_1 \rangle$, and from Eq. (10) we have
\[
|\psi_1 \rangle
= \big(\alpha | 000 \rangle + \beta | 110 \rangle\big)\dfrac{1}{\sqrt{2}} | 0 \rangle
+ \big(\alpha | 011 \rangle + \beta | 101 \rangle\big)\dfrac{1}{\sqrt{2}} | 1 \rangle. \tag{12}
\]
From the analysis of Eq. (10), it can be known that $| \psi_1 \rangle$ is an entangled state.

\textbf{(1) Distribution of quantum shares}

\vspace{0.5em}
(1.1) Alice selects the bit values of ordered pairs in the binary representation $j_1j_2\cdots j_m$ of the secret $s$, and secretly informs $P_1^{(1)},P_1^{(2)},P_1^{(4)}$ of the corresponding positions of these ordered pairs. Since the representative participants $P_1^{(1)},P_1^{(2)}$ have reconstructed the secret $s$, they can deduce $j_{i_1}j_{i_2}$, $j_{i_3}j_{i_4}$ and $j_{i_5}j_{i_6}$ from $s$. According to the system model in this section, $P_1^{(1)},P_1^{(2)},P_1^{(4)}$ can determine the Pauli operators $\sigma_1,\sigma_2,\sigma_3$ from $j_{i_1}j_{i_2}$, $j_{i_3}j_{i_4}$, $j_{i_5}j_{i_6}$ respectively, where $\sigma_i\in\{X,Y,Z\},\ i=1,2,3$.

\vspace{0.5em}
(1.2) Alice applies the operator $\sigma_1\otimes I\otimes \sigma_2\otimes \sigma_3$ sequentially to the first to fourth particles in $|\psi_1\rangle$ in left-to-right order. After the operation, we have $(\sigma_1\otimes I\otimes \sigma_2\otimes \sigma_3)|\psi_1\rangle=|\tilde{\varphi}\rangle$. Alice sends the first and second particles in $|\tilde{\varphi}\rangle$ to $P_1^{(1)}$, the third particle to $P_1^{(2)}$, and the fourth particle to $P_1^{(4)}$.

When transmitting the four particles of $|\tilde{\varphi}\rangle$ to the participants in the minimal authorized subset, it is necessary to prepare $\Gamma_E$ detection entangled states $|\varphi\rangle$ defined in Eq. (8), where each $|\varphi\rangle$ contains three particles. When distributing to participant $P_1^{(i)}$ ($i\in\{1,2,4\}$), the particles of the shared quantum state $|\tilde{\varphi}\rangle$ are interleaved among these particles from $\Gamma_E$ detection entangled states $|\varphi\rangle$.

$P_1^{(1)}$, $P_1^{(2)}$ and $P_1^{(4)}$ measure the particle states at the relevant positions informed by Alice. If their measurement results are identical (both $0$ or both $1$), they proceed to the next step; otherwise, the process will restart from (1.1).

\vspace{0.6em}
\textbf{(2) Reconstruction of the secret quantum state}
$|K\rangle=\alpha|0\rangle+\beta|1\rangle$

\vspace{0.4em}
(2.1) $P_1^{(1)}$, $P_1^{(2)}$ and $P_1^{(4)}$ receive the corresponding particles respectively. $P_1^{(1)}$ acts on its first received particle with $\sigma_1$, while $P_1^{(2)}$ and $P_1^{(4)}$ act on their respective received particles with $\sigma_2$ and $\sigma_3$. Meanwhile, $P_1^{(2)}$ and $P_1^{(4)}$ inform $P_1^{(1)}$ that they have completed their operations.

\vspace{0.4em}
(2.2) $P_1^{(1)}$ performs the CNOT operation on its first and second particles, with the first particle taken as the control qubit.

\vspace{0.4em}
The distribution of quantum shares and the recovery process of the secret quantum state $|K\rangle$ for participants on the minimal authorized subset $A_1A_2A_5$ are similar to those on $A_1A_2A_4$.

\subsection*{4.2.3 PQSS scheme for $\Gamma_3=\{A_1A_2A_4,A_1A_2A_3A_5,A_1A_3A_6\}$}
The distribution of quantum shares and the recovery process of the secret quantum state $|K\rangle$ for participants on $A_1A_2A_4$ and $A_1A_3A_6$ in $\Gamma_3$ are similar to the case of $A_1A_2A_4$ in $\Gamma_2$. The following presents the distribution of quantum shares and the recovery process of $|K\rangle$ for the minimal authorized subset $A_1A_2A_3A_5$.

Alice Takes the four-qubit entangled state for $|\varphi\rangle$ in Eq. (8), namely
\[
|\varphi\rangle=\frac{1}{\sqrt{2}}\big(|0000\rangle+|1111\rangle\big).
\]
Alice performs the CNOT operation on the first and second qubits of the direct product state $|K\rangle\otimes|\varphi\rangle$, where the first qubit is the control qubit. At this time, $|K\rangle\otimes|\varphi\rangle$ evolves into the state $|\psi_1\rangle$. According to Eq. (10), we obtain
\[
|\psi_1\rangle
=\frac{1}{\sqrt{2}}\Big[\big(\alpha|0000\rangle+\beta|1100\rangle\big)|0\rangle
+\big(\alpha|0111\rangle+\beta|1011\rangle\big)|1\rangle\Big]. \tag{13}
\]
From the analysis of Eq. (10), $|\psi_1\rangle$ is an entangled state.

\textbf{(1) Distribution of quantum shares}

(1.1) Alice selects the bit values of ordered pairs in the binary representation $j_1j_2\cdots j_m$ of the secret $s$, and secretly informs $P_1^{(1)},P_1^{(2)},P_1^{(3)},P_1^{(5)}$ of the corresponding positions of these ordered pairs. Let these ordered pairs be $j_{i_1}j_{i_2}$, $j_{i_3}j_{i_4}$, $j_{i_5}j_{i_6}$ and $j_{i_7}j_{i_8}$, respectively. Since the representative participants $P_1^(i)$ have reconstructed the secret $s$, they can obtain $j_{i_1}j_{i_2}$, $j_{i_3}j_{i_4}$, $j_{i_5}j_{i_6}$ and  $j_{i_7}j_{i_8}$ from $s$. According to the system model of this section, $P_1^{(1)},P_1^{(2)},P_1^{(3)},P_1^{(5)}$ can determine the Pauli operators $\sigma_1,\sigma_2,\sigma_3,\sigma_4$ from $j_{i_1}j_{i_2}$, $j_{i_3}j_{i_4}$, $j_{i_5}j_{i_6}$, $j_{i_7}j_{i_8}$, where $\sigma_i\in\{X,Y,Z\}$, $i\in [4]$.

\vspace{0.4em}
(1.2) Alice applies the operator $\sigma_1\otimes I\otimes \sigma_2\otimes \sigma_3\otimes \sigma_4$ sequentially to the first to fifth particles in $|\psi_1\rangle$ in left-to-right order. After the operation, we have
\[
(\sigma_1\otimes I\otimes \sigma_2\otimes \sigma_3\otimes \sigma_4)|\psi_1\rangle = |\tilde{\varphi}\rangle.
\]
Alice sends the first and second particles in $|\tilde{\varphi}\rangle$ to $P_1^{(1)}$, and the third, fourth and fifth particles to $P_1^{(2)}$, $P_1^{(3)}$ and $P_1^{(5)}$, respectively.

When transmitting the five particles of $|\tilde{\varphi}\rangle$  to the participants in the minimal authorized subset, it is necessary to prepare $\Gamma_E$ detection entangled states $|\varphi\rangle$ defined in Eq. (8), where each $|\varphi\rangle$ contains four particles. When distributing to participant $P_1^{(i)}$ with $i\in\{1,2,3,5\}$, the particles of the shared quantum state $|\tilde{\varphi}\rangle$ are interleaved among these particles from $\Gamma_E$ detection entangled states $|\varphi\rangle$.

$P_1^{(1)},P_1^{(2)},P_1^{(3)},P_1^{(5)}$ measure the particle states at the relevant positions informed by Alice. If their measurement results are identical (both $0$ or both $1$), they proceed to the next step; otherwise, the process will restart from (1.1).

\vspace{0.6em}
\textbf{(2) Reconstruction of the secret quantum state $|K\rangle = \alpha|0\rangle + \beta|1\rangle$}

\vspace{0.4em}
(2.1) $P_1^{(1)},P_1^{(2)},P_1^{(3)},P_1^{(5)}$ receive the corresponding particles respectively. $P_1^{(1)}$ acts on its first received particle with $\sigma_1$, while $P_1^{(2)},P_1^{(3)},P_1^{(5)}$ act on their respective received particles with $\sigma_2$, $\sigma_3$ and $\sigma_4$. Meanwhile, $P_1^{(2)},P_1^{(3)},P_1^{(5)}$ inform $P_1^{(1)}$ that they have completed their operations.

\vspace{0.4em}
(2.2) $P_1^{(1)}$ performs the CNOT operation on its first and second particles, where the first particle is taken as the control qubit.

\subsection*{4.2.4 QSS Scheme for $\Gamma_4=\{A_1A_2A_4A_7,A_1A_2A_3A_5,A_1A_3A_6A_7\}$}
The distribution of quantum shares and the recovery process of the quantum secret state $|K\rangle$ for participants of the three minimal authorized subsets in $\Gamma_4$ are similar to those of $A_1A_2A_3A_5$ in $\Gamma_3$, so it will not be elaborated here.

\begin{figure}[htbp]
	\centering
	\includegraphics[width=0.9\textwidth]{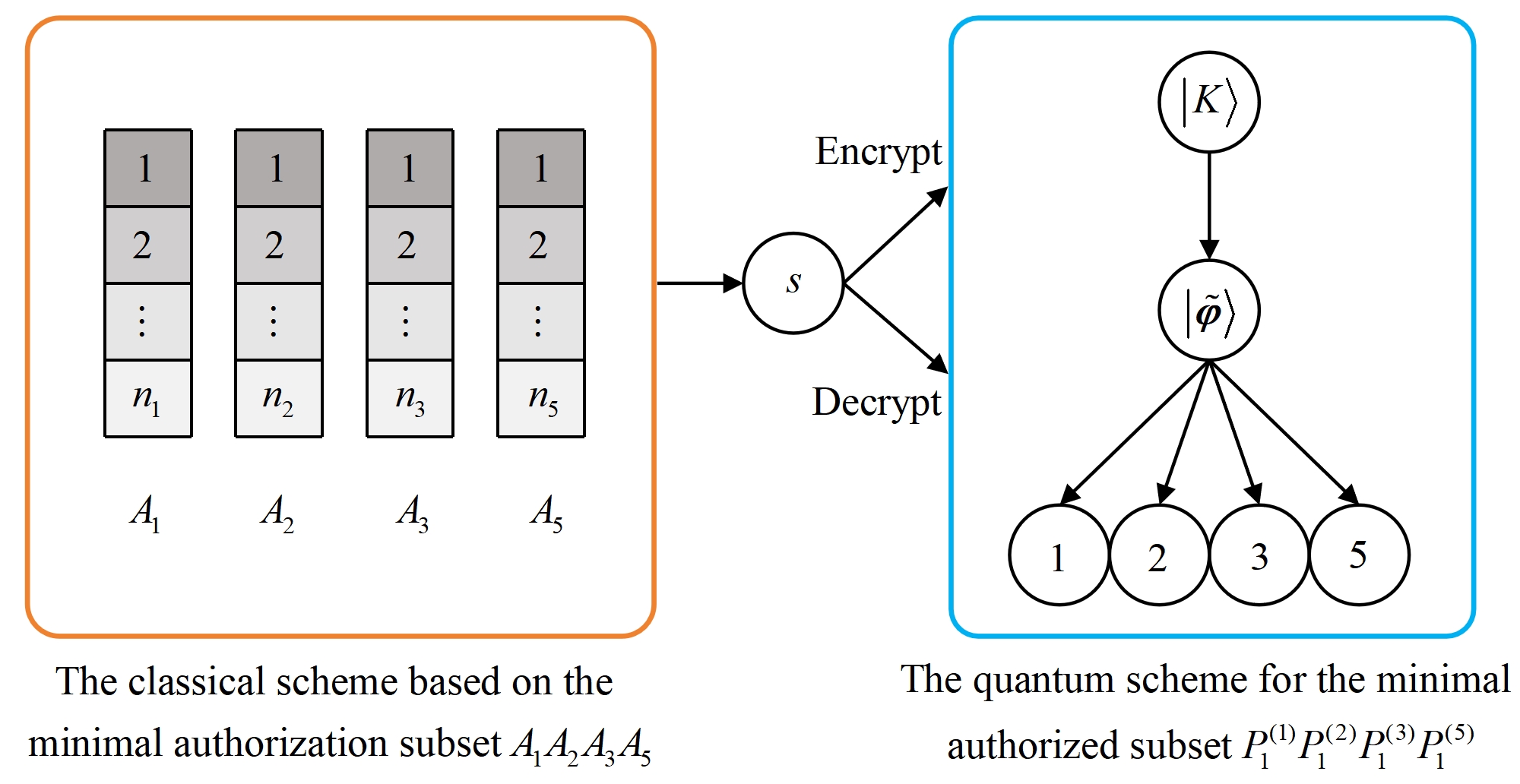}
	\caption{Flowchart of the proposed scheme}
	\label{fig:flowchart}
\end{figure}

Fig.1 presents an example of the PQSS scheme for the minimal authorized subset $A_1A_2A_3A_5$ based on the access structure $\Gamma_3$. The particles from entangled state  $|\tilde{\varphi}\rangle$  are then distributed to the representative participants $P_1^{(1)}P_1^{(2)}P_1^{(3)}P_1^{(5)}$  in the QAS $R_3$. The quantum secret $|K\rangle$ is encrypted by Pauli operators corresponding to classical keys $s$.  The representatives $P_1^{(1)}P_1^{(2)}P_1^{(3)}P_1^{(5)}$ can finally decrypt and recover the quantum secret state $|K\rangle$ by performing Pauli operator operations with the classical secret $s$.

\subsection{Correctness proof of this scheme}
\textbf{Theorem 4.1} The scheme constructed on the minimal access structure $\Gamma_i( i\in[4])$ is a PQSS scheme. That is, any minimal authorized subset in $\Gamma_i$ can recover the secret quantum state $|K\rangle$, while any unauthorized subset cannot obtain any information about the secret.

\vspace{0.4em}
\textbf{Proof}: For convenience of description, we first prove the correctness of the scheme constructed on the minimal access structure $\Gamma_1$. We first analyze the recovery of the quantum state $|K\rangle$ for the minimal authorized subset $A_1A_2$. Participants in $A_1A_2$ have obtained the classical key $s$ by adopting the scheme in Lemma 2.4, where $s\in F_p$. Therefore, in the representative access structure $R_1$, the two participants $P_1^{(1)}$ and $P_1^{(2)}$ in the minimal authorized subset $A_1A_3$ can obtain the classical key $s$.

Furthermore, participants $P_1^{(1)}$ and $P_1^{(2)}$ in the minimal authorized subset $A_1A_2$ use the classical key $s$ to decrypt the quantum state encrypted by Alice. Following (2) in Section 4.2.1, they can cooperatively recover the secret quantum state $|K\rangle$. This indicates that representative participants of minimal authorized subset from $R_1$ can finally obtain the secret quantum state $|K\rangle$. Similarly, it can be proved that participants in the minimal authorized subsets $A_1A_3$ and $A_1A_4$ can each recover the secret quantum state $|K\rangle$.

On the other hand, since the classical scheme based on $\Gamma_1$ is a perfect secret sharing scheme, any unauthorized subset not belonging to $\Gamma_1$ cannot obtain the secret $s$. Furthermore, even if participants in an unauthorized subset intercept partial information of the quantum shares during the quantum share distribution phase, they cannot obtain all quantum shares due to eavesdropping detection. Without the secret $s$, they are unable to decrypt the encrypted quantum state $|\tilde{\varphi}\rangle$. Hence, participants in any unauthorized subset of $\Gamma_1$ cannot acquire the secret quantum state $|K\rangle$.

The correctness proof of the schemes constructed on the minimal access structures $\Gamma_i$ ($i\in [2,4]$) is similar to that for $\Gamma_1$. Therefore, the proposed scheme is a PQSS scheme on the QAS $\Gamma_i( i\in[4])$.

\section{Conclusions}
Developing efficient QSS schemes remains a prominent research focus in quantum communication \cite{ref33,ref34}. Among these, single-photon schemes and quantum entangled-state schemes represent the two primary technical approaches for achieving efficient QSS. Single-photon technology effectively reduces the physical implementation complexity, enhances system stability and reliability, and overcomes inherent bottlenecks in entangled-state preparation and distribution \cite{ref35,ref36}; whereas entangled-state technology focuses on quantum resource innovation and topological structure optimization \cite{ref37}.

This paper mainly investigates the judgment criteria for non-maximal quantum access structures and the construction methods of efficient quantum secret sharing schemes for such structures based on entangled states, with a special focus on hyperstar quantum access structures containing three hyperedges. Prior studies have designed efficient perfect quantum secret sharing schemes by combining single photons with classical secret sharing theory. On this basis, this paper aims to propose a general construction method based on entangled states. This category of non-maximal quantum access structures neither belongs to ramp-type quantum access structures nor can be generated by stabilizer codes, so the existing efficient scheme construction methods applicable to the above two types of structures cannot be directly adopted. This paper first defines the representative-element quantum access structure of such structures as a substructure of hyperstar quantum access structures with three hyperedges. We construct an imperfect quantum secret sharing scheme over this substructure using entangled states, and introduce a classical secret sharing scheme corresponding to the three-hyperedge hyperstar quantum access structure as an auxiliary component to further develop a complete perfect quantum secret sharing scheme. Since quantum resources are only lightly deployed for the representative-element access structure, the proposed scheme achieves both high efficiency and practicality.

The method proposed in this paper can provide important theoretical references and guidance for the design and construction of efficient quantum secret sharing schemes under general quantum access structures, and further offer new theoretical support and implementation paths for the practical deployment of secret sharing technologies in asymmetric quantum networks.
\section*{Acknowledgment}
This work was supported by the National Natural Science Foundation of China 12201484.

\end{document}